\documentclass[1p,final]{elsarticle}
\usepackage{amsfonts,color,morefloats,pslatex}
\usepackage{amssymb,amsthm, amsmath,latexsym}

\allowdisplaybreaks[4]

\newtheorem{rem}{Remark}
\newtheorem{theorem}{Theorem}
\newtheorem{lemma}[theorem]{Lemma}

\newtheorem{corollary}[theorem]{Corollary}

\newtheorem{example}{Example}

\newcommand{\tr}{{\mathrm{Tr}}}

\newcommand{\gf}{{\mathrm{GF}}}

\newcommand{\cP}{{\mathcal{P}}}
\newcommand{\cB}{{\mathcal{B}}}

\newcommand{\C}{{\mathcal{C}}}

\newcommand{\bD}{{\mathbb{D}}}

\begin{document}

\begin{frontmatter}



\title{ Some $3$-designs and shortened codes from binary cyclic codes with three zeros
\tnotetext[fn1]{C. Xiang's research was supported by
the National Natural Science Foundation of China under grant numbers 12171162 and 11971175, and the Basic Research Project of Science and Technology Plan of Guangzhou city of China  under grant number 202102020888; C. Tang's research was supported by the National Natural Science Foundation of China under grant number 12231015, the Sichuan Provincial Youth Science and Technology Fund under grant number 2022JDJQ0041 and  the Innovation Team Funds of China West Normal University under grant number KCXTD2022-5.}}


\author[cx]{Can Xiang}
\address[cx]{College of Mathematics and Informatics, South China Agricultural University, Guangzhou, Guangdong 510642, China}
\ead{cxiangcxiang@hotmail.com}
\author[cmt]{Chunming Tang}
\address[cmt]{School of Information Science and Technology, Southwest Jiaotong University, Chengdu 610031, China}
\ead{tangchunmingmath@163.com}



\begin{abstract}
Linear codes and $t$-designs are interactive with each other. It is well known that some $t$-designs have been
constructed by using certain linear codes in recent years. However, only
a small number of  infinite families of the extended codes of linear codes holding an infinite family of $t$-designs with $t\geq 3$ are reported in
the literature. In this paper, we study the extended codes of the augmented codes of a class of binary cyclic codes with three zeros and their dual codes, and show that those codes hold $3$-designs. Furthermore, we obtain some shortened codes from the studied cyclic codes and explicitly determine their parameters. Some of those shortened codes are optimal or almost optimal.

\end{abstract}


\end{frontmatter}

\section{Introduction}
Let $p$ be a prime and $q = p^m$ for some positive integer $m$. Let $\gf(q)$ be the finite field
of cardinality $q$. A $[v,\, k,\,d]$ linear code $\C$ over $\gf(q)$ is a $k$-dimensional subspace of $\gf(q)^v$ with minimum (Hamming) distance $d$.
A $[v,\, k,\,d]$ linear code $\C$ is said to be {\em cyclic} if
$(c_0,c_1, \cdots, c_{v-1}) \in \C$ implies $(c_{v-1}, c_0, c_1, \cdots, c_{v-2}) \in \C$.
It is known that a cyclic code is a special linear code. Although the error correcting capability of cyclic codes may not be as good as some other linear codes in general, cyclic codes have wide applications in communication and storage systems as they have efficient encoding and decoding algorithms \cite{Chien5,Forney12,Prange28}. Thus, cyclic codes have attracted much attention in coding theory and a lot of progress has been made (see, for example, \cite{dinghell2013,Ding2018,sihem1,YZD2018,zhou20131}).

It is known that linear codes and $t$-designs are closely related. A $t$-design can be induced to a linear code (see, for example, \cite{Dingtv2020,Dingt20201}) and a linear code $\C$ may induce a $t$-design under certain conditions. As far as we know, a lot of $2$-designs and $3$-designs have been constructed from some special linear codes (see, for example, \cite{Ding18dcc,Ding18jcd,ding2018,Tangdcc2019}). Recently, an infinity family of linear codes holding $4$-designs was constructed in \cite{Tangding2020}. It remains open if there is an infinite family of linear codes holding $5$-designs. In fact, only a few infinite families of the extended codes of linear codes holding an infinite family of $3$-designs are reported in the literature. Motivated by this fact, we will consider a class of binary cyclic codes
\begin{eqnarray}\label{ce}
\C^{(e)}=\left\{\left(\tr(ax^{2^{2e}+1}+bx^{2^e+1}+cx)_{x \in \gf(q)^*}\right):a,b,c\in \gf(q)\right\}.
\end{eqnarray}
with zero set $\{1,2^e+1,2^{2e}+1\}$, where $q=2^m$, $\gf(q)^*=\gf(q)\setminus \{0\}$, $e$ is a positive integer with $1\leq e\leq m-1$ and $e\notin \{\frac{m}{3},\frac{2m}{3}\}$, $m/\gcd(m,e)$ is odd and $\tr$ is the trace function from $\gf(q)$ to $\gf(2)$. The first objective of this paper is to show that the extended code  of the augmented code  of $\C^{(e)}$ and its dual code  hold $3$-designs and determine their parameters.

It is worth noting that the shortening and puncturing technologies are two important approaches to constructing new linear codes. In 2019, Tang et al. obtained some ternary linear codes with few weights by shortening and puncturing a class of ternary codes in \cite{Tangdcc2019}. Afterward, they also presented a general theory for punctured and shortened codes
of linear codes supporting t-designs and generalized Assmus-Mattson theorem in \cite{Tangit2019}. Very recently, Liu, Ding and Tang \cite{LDT2020} proved
some general theories for shortened linear codes and studied some shortened codes of known special codes such as Hamming codes, Simplex codes, some Reed-Muller codes and ovoid codes. Meanwhile, Xiang, Tang and Ding\cite{XTD2020} obtained some shortened codes of linear codes from almost perfect nonlinear and  perfect nonlinear functions and determined their parameters.
However, till now little research on the shortening technique has been done
and there are only a handful references on shortened linear codes, and it is in general hard to determine their weight distributions. Based on this fact,
the second objective of this paper is to study some shortened codes of $\C^{(e)}$ and determine their parameters. Some of these shortened codes presented in this paper are optimal or almost optimal.

The rest of this paper is arranged as follows. Section \ref{sec-pre} states some notation and results
about linear codes, combinatorial $t$-designs and exponential sums. Some infinite families of $3$-designs are presented in Section \ref{sec-des3}.
Section \ref{sec-main} gives some shortened codes and determines their parameters. The conclusion of this paper is given in Section  \ref{sec-summary}.

\section{Preliminaries}\label{sec-pre}


In this section, some notation and basic facts are described and will be needed later.

\subsection{Some results of linear codes}

Let $\C$ be a $[v,k,d]$ linear code over $\gf(q)$. We call $\C$ \emph{distance-optimal} if no $[v,k,d+1]$ code exists and \emph{dimension-optimal} if no $[v,k+1,d]$ code exists.  $\C$ is said to be \emph{length-optimal} if there is no $[v',k,d]$ code exists with $v' < v$. A code is said to be \emph{optimal} if it is distance-optimal, or dimension-optimal, or length-optimal, or meets a
bound for linear codes. A $[v,k,d]$ code is said to be \emph{almost optimal} if a $[v,k+1,d]$, or $[v,k,d+1]$, or $[v-1,k,d]$ code is optimal. The augmented code of $\C$ is denoted by $\widehat{\C}$ with generator matrix
$$\begin{bmatrix} ~G~ \\ ~\textbf{1}~  \end{bmatrix} ,$$
where $\textbf{1}=(1,1,\cdots,1)$ and $G$ is the generator matrix of $\C$. For any codeword $c=(c_0,c_1,\cdots,c_{v-1})\in \C$, we extend it into the vector
$$
\bar{c}=(c_0,c_1,\cdots,c_{v-1}, c_v),
$$
where $c_v=-(c_0+c_1+\ldots + c_{v-1})$. The extended code $\overline{\C}$ of $\C$
is then defined by
$$
\overline{\C}=\{\bar{c}: c \in \C \}.
$$
 If $H$ is the parity check matrix of $\C$, then the parity check matrix of $\overline{\C}$ is
$$\begin{bmatrix}\textbf{1} &1\\H& \textbf{0} \end{bmatrix} ,$$
where $\textbf{1}=(1,1,\cdots,1)$ and $\textbf{0}=(0,0,\cdots,0)^\top$.

Let $A_i(\C)$ denote the number of codewords with Hamming weight $i$ in
$\C$. The {\em weight enumerator} of $\C$ is defined by
$
1+\sum_{i=1}^{v} A_i(\C) z^i.
$
The weight enumerator of a linear code contains significant
information including its error correcting capability and the
error probability of error detection. Thus much work focuses on the determination of
the weight distributions of linear codes (see, for example, \cite{ding2018,Ding16, DingDing2, sihem2020,sihem2017,tang2016,WZ2020,zhou20131} and the references therein). A code $\C$ is said to be a $t'$-weight code if the number of nonzero
$A_i(\C)$ in the sequence $(A_1(\C), A_2(\C), \cdots, A_v(\C))$ is equal to $t'$.
Denote by $\C^\bot$  and $(A_0(\C^{\perp}), A_1(\C^{\perp}), \dots, A_\nu(\C^{\perp}))$ the dual code of a linear code $\C$ and its weight distribution, respectively.
The \emph{Pless power moments} \cite{HP10}, i.e.,
\begin{align}\label{eq:PPM}
 \sum_{i=0}^\nu  i^t A_i(\C)= \sum_{i=0}^t (-1)^i A_i(\C^{\perp})   \left [  \sum_{j=i}^t   j ! S(t,j) q^{k-j} (q-1)^{j-i} \binom{\nu-i}{\nu -j}  \right ],
 \end{align}
play an important role in calculating
the weight distributions of linear codes, where $A_0(\C)=1$, $0\le t \le \nu$ and $S(t,j)=\frac{1}{j!} \sum_{i=0}^j  (-1)^{j-i} \binom{j}{i} i^t$.

Let $T$ be a set of $t$ coordinate positions in $\C$. We puncture $\mathcal  C$  on
$T$ and obtain a linear code which is called the \emph{punctured code}  of $\mathcal C$ on $T$ and denoted by $\mathcal  C^T$.
We use $\mathcal C(T)$ to denote the set of codewords that are
$\mathbf{0}$ on $T$.
We now puncture $\mathcal C(T)$ on $T$, and obtain a linear code $\mathcal C_{T}$, which is called the \emph{shortened code} of $\mathcal C$ on $T$.
The following property plays an important role in determining the parameters of the punctured and shortened codes of $\mathcal{C}$  in \cite[Theorem 1.5.7]{HP10}.
\begin{lemma}\cite{HP10}\label{lem:C-S-P}
Let $\C$ be a $[\nu,k,d]$ linear code over $\gf(q)$ and  $d^{\perp}$  the minimum distance of $\mathcal  C^{\perp}$.
Let $T$ be any set of $t$ coordinate positions. Then
\begin{itemize}
  \item $\left ( \mathcal C_{T} \right )^{\perp} = \left ( \mathcal C^{\perp}  \right)^T$ and $\left ( \mathcal C^{T} \right )^{\perp} = \left ( \mathcal C^{\perp}  \right)_T$.
  \item If $t<\min \{d, d^{\perp} \}$, then the codes $\mathcal C_{T}$ and $\mathcal C^T$ have  dimension  $k-t$  and $k$, respectively.
\end{itemize}
\end{lemma}

\subsection{Combinatorial t-designs and some related results}

Let $k$, $t$ and $v$ be positive integers with $1 \leq t \leq k \leq  v$. Let $\cP$ be a set with $v$ elements and $\cB$ be a set of some $k$-subsets of
$\cP$. $\cB$  is called the point set and  $\cP$ is called the block set in general. The incidence structure
$\bD = (\cP, \cB)$ is called a $t$-$(v, k, \lambda)$ {\em design\index{design}} (or {\em $t$-design\index{$t$-design}}) if each $t$-subset of $\cP$ is contained in exactly $\lambda$ blocks of
$\cB$.
Let $\binom{\cP}{k}$ denote the set consisting of all $k$-subsets of the point set $\cP$. Then the incidence structure $(\cP, \binom{\cP}{k})$ is a $k$-$(v, k, 1)$ design and is called a \emph{complete design}. The special incidence structure $(\cP, \emptyset)$ is called a $t$-$(v, k, 0)$ trivial design
for all $t$  and $k$ . A combinatorial $t$-design is said to be {\em simple\index{simple}} if its block set $\cB$ does not have
a repeated block. When $t \geq 2$ and $\lambda=1$, a $t$-$(v,k,\lambda)$ design is called  a
{\em Steiner system\index{Steiner system}}
and denoted by $S(t,k, v)$. The parameters of a combinatorial $t$-$(v, k, \lambda)$ design must satisfy the following equation:
\begin{eqnarray}\label{eq:bb}
b  =\lambda \frac{\binom{v}{t}}{\binom{k}{t}}
\end{eqnarray}
where $b$ is the cardinality of $\cB$.


It is well known that $t$-designs and linear codes  are interactive with each other.
A $t$-design $\mathbb  D=(\mathcal P, \mathcal B)$ can be used to construct a linear code over GF($q$) for any $q$ (see, for example, \cite{Dingt20201,ton1,ton2}).
Meanwhile, a linear code $\C$ may produce a $t$-design which is formed by supports of codewords of a fixed Hamming weight in $\C$. Let $\nu$ be the length of $\mathcal C$ and the set of the coordinates of codewords in $\mathcal C$ is denoted by $\mathcal P(\mathcal C)=\{0,1, 2, \dots, \nu-1\}$. The \emph{support} of  $\mathbf c$
is defined by
\begin{align*}
\mathrm{Supp}(\mathbf c) = \{i: c_i \neq 0, i \in \mathcal P(\mathcal C)\}
\end{align*}
for any codeword $\mathbf c =(c_0, c_1, \dots, c_{\nu-1})$ in $\mathcal C$.
Let $\mathcal B_{w}(\mathcal C)$ denote the set $\{\{   \mathrm{Supp}(\mathbf c): wt(\mathbf{c})=w
~\text{and}~\mathbf{c}\in \mathcal{C}\}\}$, where $\{\{\}\}$ is the multiset notation. For some special code $\mathcal C$,
the incidence structure $\left (\mathcal P(\mathcal C),  \mathcal B_{w}(\mathcal C) \right)$
could be a $t$-$(v,w,\lambda)$ design for some positive integers $t$ and $\lambda$.
We say that the code $\mathcal C$ \emph{supports $t$-designs} if $\left (\mathcal P(\mathcal C),  \mathcal B_{w}(\mathcal C) \right)$ is a $t$-design for all $w$ with $0\le w \le \nu$. By definition, such design
$\left (\mathcal P(\mathcal C),  \mathcal B_{w}(\mathcal C) \right)$ could have some repeated
blocks, or could be simple, or may be trivial.
In this way, many $t$-designs have been constructed from linear codes (see, for example, \cite{Ding18dcc,Ding18jcd,ding2018,Tangding2020,du1,Tangdcc2019}). A major way to construct combinatorial $t$-designs with linear codes over finite fields is the use of linear codes with $t$-transitive or $t$-homogeneous automorphism groups (see \cite[Theorem 4.18]{ding2018}) and some combinatorial $t$-designs (see, for example, \cite{LiuDing2017,Liudingtang2021}) were obtained by this way.
Another major way to construct $t$-designs with linear codes is the use of the
Assmus-Mattson Theorem (AM Theorem for short) in \cite[Theorem 4.14]{ding2018} and the generalized version of the
AM Theorem in \cite{Tangit2019}, which was recently employed to construct a number of $t$-designs (see, for example, \cite{ding2018,du1}).
The following theorem is a generalized version of the
AM Theorem, which was developed in \cite{Tangit2019} and will be needed in this paper.

\begin{theorem}\cite{Tangit2019}\label{thm-designGAMtheorem}
Let $\mathcal C$ be a linear code over $\mathrm{GF}(q)$ with minimum distance $d$ and length $\nu$.
Let $\mathcal C^{\perp}$ denote the dual of $\mathcal C$ with minimum distance $d^{\perp}$.
Let $s$ and $t$ be positive integers with $t< \min \{d, d^{\perp}\}$. Let $S$ be a $s$-subset
of the set $\{d, d+1, d+2, \ldots, \nu-t  \}$.
Suppose that
$\left ( \mathcal P(\mathcal C), \mathcal B_{\ell}(\mathcal C) \right )$ and $\left ( \mathcal P(\mathcal C^{\perp}), \mathcal B_{\ell^{\perp}}(\mathcal C^{\perp}) \right )$
are $t$-designs  for
$\ell    \in \{d, d+1, d+2, \ldots, \nu-t  \} \setminus S $ and $0\le \ell^{\perp} \le s+t-1$, respectively. Then
the incidence structures
 $\left ( \mathcal P(\mathcal C) , \mathcal B_k(\mathcal C) \right )$ and
  $\left ( \mathcal P(\mathcal C^{\perp}), \mathcal B_{k}(\mathcal C^{\perp}) \right )$ are
  $t$-designs for any
$t\le k  \le \nu$, and particularly,
\begin{itemize}
\item the incidence structure $\left ( \mathcal P(\mathcal C) , \mathcal B_k(\mathcal C) \right )$ is a simple $t$-design
      for all integers $k$ with $d \leq k \leq w$, where $w$ is defined to be the largest  integer
      such that $w \leq \nu$ and
      $$
      w-\left\lfloor \frac{w+q-2}{q-1} \right\rfloor <d;
      $$
\item  and the incidence structure $\left ( \mathcal P(\mathcal C^{\perp}), \mathcal B_{k}(\mathcal C^{\perp}) \right )$ is
       a simple $t$-design
      for all integers $k$ with $d \leq k \leq w^\perp$, where $w^\perp$ is defined to be the largest integer
      such that $w^\perp \leq \nu$ and
      $$
      w^\perp-\left\lfloor \frac{w^\perp+q-2}{q-1} \right\rfloor <d^\perp.
      $$
\end{itemize}
\end{theorem}

\subsection{Exponential sums}

Let $q=2^m$ and $\tr$ be the trace function from $\gf(2^m)$ to
$\gf(2)$ in the rest of this paper.
Let $(a,\ b,\ c )\in \gf(q)^3$ and  define the following exponential sums,
\begin{equation*}\label{defkm}
K_m(a) = \sum \limits_{x\in \gf(q)^*}(-1)^{\tr(ax+x^{-1})},
\end{equation*}
\begin{equation*}\label{defcm}
C_m(a,b) = \sum \limits_{x\in \gf(q)}(-1)^{\tr(ax^3+bx)},
\end{equation*}
\begin{equation*}\label{defgm}
G_m(a,b) = \sum \limits_{x\in \gf(q)^*}(-1)^{\tr(ax^3+bx^{-1})}.
\end{equation*}

Let $(a,b,c)\in \gf(q)^3$ and
$N_{(a,b,c)}$ denote the number of $\{x,y,z,u\}\subseteq \gf(q)$ satisfying the system of equations:
\begin{eqnarray*}
\left\{
\begin{array}{ll}
  x+y+z+u=a&\\[2mm]
 x^{3}+y^{3}+z^{3}+u^3=b&\\[2mm]
x^{5}+y^{5}+z^{5}+u^5=c
 \end{array}
 \right..
\end{eqnarray*}
Then we have the following results which was described in \cite{Johansen2009}.
\begin{lemma}\label{lem-solution}\cite{Johansen2009}
Let $m\geq 1$ be odd, $a\in \gf(q)^*$, $(b,c)\in \gf(q)^2$, $\varepsilon=\tr(\frac{b}{a^3})$ and
$$
\mu=\frac{c}{a^5}+\frac{b^2}{a^6}+\frac{b}{a^3}.
$$
Then
\begin{eqnarray}\label{eqn-solution}
N_{(a,b,c)}
&=&
\left\{
\begin{array}{ll}
0              &   \mbox{ if } \mu=1, \\ \nonumber
\frac{1}{24} \left(2^m-5+3G_m(\mu+1,\mu+1)\right)+\frac{1}{12}(-1)^{\varepsilon+1}    &   \mbox{ if } \mu \neq 1 .\\
~~~~~~~~\times \left(K_m(\mu+1)+C_m(1,(\mu+1)^{1/3})-3\right)    &    ~\\
\end{array} \right.
\end{eqnarray}
\end{lemma}

\section{Infinite families of $3$-designs from cyclic codes} \label{sec-des3}

Our task in this section is to study the extended code $\overline{\widehat{\C^{(e)}}}$ of the augmented code $\widehat{\C^{(e)}}$ of the cyclic code $\C^{(e)}$  and its dual $\overline{\widehat{\C^{(e)}}}^\perp$, and prove that these codes hold $3$-designs.

By definitions, it follows that the trace expression of $\overline{\widehat{\C^{(e)}}}$  is given by
\begin{eqnarray}\label{eq:trace}
\overline{\widehat{\C^{(e)}}}=\left\{\left(\tr(ax^{2^{2e}+1}+bx^{2^e+1}+cx+h)_{x \in \gf(q)}\right):a,b,c,h\in \gf(q)\right\},
\end{eqnarray}
where $q=2^m$, $e$ is a positive integer with $1\leq e\leq m-1$ and $e\notin \{\frac{m}{3},\frac{2m}{3}\}$, $m/\gcd(m,e)$ is odd. Note that the code $\overline{\widehat{\C^{(e)}}}$ is affine-invariant \cite{dingtang2020ccds}.

It is well known that the parameters of $\C^{(e)}$ is determined by Luo \cite{LJQ2010}, i.e., the code $\C^{(e)}$ defined by (\ref{ce}) is a $[q-1, 3m]$ linear codes with the weight distribution in Table \ref{tab-31}. The following result can be easily obtained from the parameters of $\C^{(e)}$  and we omit its proof.

\begin{table}[ht]
\begin{center}
\caption{The weight distribution of $\C^{(e)}$ for $ m/ \gcd(m,e)$ odd, where $d=\gcd(m,e)$.} \label{tab-31}
\begin{tabular}{|c|c|}
\hline
weight & multiplicity \\[2mm]
\hline
$0$ & $1$
\\[2mm]
\hline
$2^{m-1}-2^{(m+3d-2)/2}$ & $\frac{(2^{m-3d-1}+2^{(m-3d-2)/2})(2^{m-d}-1)(2^m-1)}{2^{2d}-1}$
\\[2mm]
\hline
$2^{m-1}-2^{(m+d-2)/2}$ & $\frac{(2^{m-d-1}+2^{(m-d-2)/2})(2^m-1)(2^{m+2d}-2^{m}-2^{m-d}+2^{2d})}{2^{2d}-1}$
\\[2mm]
\hline
$2^{m-1}$ & $\scriptstyle{(2^m-1)(2^{2m}-2^{2m-d}+2^{2m-4d}+2^m-2^{m-d}-2^{m-3d}+1)}$
\\[2mm]
\hline
$2^{m-1}+2^{(m+d-2)/2}$ & $\frac{(2^{m-d-1}-2^{(m-d-2)/2})(2^m-1)(2^{m+2d}-2^{m}-2^{m-d}+2^{2d})}{2^{2d}-1}$
\\[2mm]
\hline
$2^{m-1}+2^{(m+3d-2)/2}$& $\frac{(2^{m-3d-1}-2^{(m-3d-2)/2})(2^{m-d}-1)(2^m-1)}{2^{2d}-1}$
\\[2mm]
\hline
\end{tabular}
\end{center}
\end{table}

\begin{lemma}\label{lem-extend}
Let $m$ and $e$ be positive integers with $m\geq 5$ and  $1\leq e\leq m-1$, $e\notin \{\frac{m}{3},\frac{2m}{3}\}$ and $m/\gcd(m,e)$ odd. Let $\C=\C^{(e)}$ be defined by (\ref{ce}).
Then the extended code $\overline{\widehat{\C}}$  has parameters $[q, 3m+1,2^{m-1}-2^{(m+3d-2)/2}]$  and  the weight distribution in Table \ref{tab-extend31}, where $d=\gcd(m,e)$.
\end{lemma}
\begin{table}[ht]
\begin{center}
\caption{The weight distribution of $\overline{\widehat{\C^{(e)}}}$  for $ m/ \gcd(m,e)$ odd, where $d=\gcd(m,e)$.} \label{tab-extend31}
\begin{tabular}{|c|c|}
\hline
weight & multiplicity \\[2mm]
\hline
$0$ & $1$
\\[2mm]
\hline
$2^{m-1}-2^{(m+3d-2)/2}$ & $\frac{(2^{m-3d})(2^{m-d}-1)(2^m-1)}{2^{2d}-1}$
\\[2mm]
\hline
$2^{m-1}-2^{(m+d-2)/2}$ & $\frac{(2^{m-d})(2^m-1)(2^{m+2d}-2^{m}-2^{m-d}+2^{2d})}{2^{2d}-1}$
\\[2mm]
\hline
$2^{m-1}$ & $\scriptstyle{2(2^m-1)(2^{2m}-2^{2m-d}+2^{2m-4d}+2^m-2^{m-d}-2^{m-3d}+1)}$
\\[2mm]
\hline
$2^{m-1}+2^{(m+d-2)/2}$ & $\frac{(2^{m-d})(2^m-1)(2^{m+2d}-2^{m}-2^{m-d}+2^{2d})}{2^{2d}-1}$
\\[2mm]
\hline
$2^{m-1}+2^{(m+3d-2)/2}$& $\frac{(2^{m-3d})(2^{m-d}-1)(2^m-1)}{2^{2d}-1}$
\\[2mm]
\hline
$2^m$                & $1$
\\
\hline
\end{tabular}
\end{center}
\end{table}

\begin{example}\label{exam-0101}
Let $(m,e)=(7,1)$. Then the code $\overline{\widehat{\C}}$ is a $[128,22,48]$ binary linear code with the weight enumerator $1+42672 z^{48}+877824 z^{56}+ 2353310 z^{64}+ 877824 z^{72}+42672 z^{80}+ z^{128}$. The dual code $\overline{\widehat{\C}}^\perp$ of $\overline{\widehat{\C}}$  has parameters $[128,106,8]$. Both $\overline{\widehat{\C}}$ and $\overline{\widehat{\C}}^\perp$ are optimal according to the tables of best known codes maintained at http: //www.codetables.de.
\end{example}

\begin{example}\label{exam-0102}
Let $(m,e)=(10,2)$. Then the code $\overline{\widehat{\C}}$ is a $[1024,31,384]$ binary linear code with the weight enumerator
$$1+278256 z^{384}+263983104 z^{480}+ 1618960926 z^{512}+ 263983104 z^{544}+278256 z^{640}+ z^{1024}.$$
\end{example}

Next we give the parameters of the dual code of $\overline{\widehat{\C^{(e)}}}$, which  will be employed later.

\begin{lemma}\label{lem-extenddual}
Let symbols and notation be the same as before. Let $m$ and $e$ be positive integers with $m\geq 5$ and  $1\leq e\leq m-1$, $e\notin \{\frac{m}{3},\frac{2m}{3}\}$ and $m/\gcd(m,e)$ odd. Let $\C=\C^{(e)}$ be defined by (\ref{ce}).
Denote $d=\gcd(m,e)$. Then the dual code $\overline{\widehat{\C}}^{~\perp}$  of $\overline{\widehat{\C}}$ has parameters $[q, q-3m-1,d']$, where
\begin{eqnarray*}
d'=
\left\{
\begin{array}{ll}
8,& \mbox{if $d=1$}\\[2mm]
4 ,& otherwise
\end{array}
 \right..
\end{eqnarray*}
In particular, the number of the minimum weight codewords in $\overline{\widehat{\C}}^{~\perp}$ is
\begin{align}\label{eq-A8}
A_8 (\overline{\widehat{\C}}^\perp)= \frac{1}{315} \cdot  2^{m-7} (2^m-1) & (-272 + 39\cdot  2^{2 + m} - 3 \cdot 4^{m+1} + 8^{m}),
\end{align}
when $d=\gcd(m,e)=1$.
\end{lemma}

\begin{proof}
By Lemma \ref{lem-extend}, it is obvious that the code $\overline{\widehat{\C}}^{~\bot}$ has length $q$ and dimension $q-3m-1$. From the weight distribution in Table \ref{tab-extend31} and the first nine Pless power moments in (\ref{eq:PPM}), we obtain that $A_i(\overline{\widehat{\C}}^\perp)=0$ for any $i\in \{1,2,3,5,7\}$  and

\begin{eqnarray*}
A_i(\overline{\widehat{\C}}^\perp)=
\left\{
\begin{array}{ll}
\frac{1}{24}(2^{m}-1)(2^{d + m} - 2^{1 +  m}), &  \mbox{if $i=4$}\\[2mm]
-\frac{1}{45}2^{m-4} (2^d-2) (-32+2^{2 + d} + 2^{1 + 2 d} + 8^d - 2^m (3 \cdot 2^d + 4^d-8)),     &  \mbox{if $i=6$}  \\   [2mm]
\frac{1}{315}
\cdot  2^{m-7} (2^m-1)  (8^m - 16 (213 - 77\cdot 2^{1 + d} + 7 \cdot 16^d) \\
+ 4^m (-132 + 91\cdot 2^{1 + d} - 27 \cdot 4^d - 27 \cdot 8^d + 16^d + 64^d)    \\
-2^m (-1380 + 357 \cdot 2^{2 + d} - 7 \cdot 2^{4 + 3 d} - 7 \cdot 4^{2 + d} - 7 \cdot 4^{1 + 2 d} + 32^d + 128^d)).
& \mbox{if $i=8$}
\end{array}
\right.
\end{eqnarray*}
When $d=1$, it is clear that $A_4(\overline{\widehat{\C}}^\perp)=A_6(\overline{\widehat{\C}}^\perp)=0$, $A_8(\overline{\widehat{\C}}^\perp)> 0$ and Equation (\ref{eq-A8}) follows. When $d\neq1$, from the above system of equations we have that $A_4(\overline{\widehat{\C}}^\perp)> 0$. This completes the proof.
\end{proof}

\begin{example}\label{exam-01}
Let $(m,e)=(5,2)$. Then $\overline{\widehat{\C}}$  and $\overline{\widehat{\C}}^\perp$ have the same parameters $[32,16,8]$ and the same weight enumerator
$1 + 620 z^{8} +  13888 z^{12} +  36518  z^{16}+ 13888 z^{20}+ 620 z^{24}+ z^{32}.$  Both $\overline{\widehat{\C}}$  and $\overline{\widehat{\C}}^\perp$ are optimal according to the tables of best known codes  maintained at http://www.codetables.de.
The number of the codewords of the minimum weight $8$ in $\overline{\widehat{\C}}^\perp$ is 620.
\end{example}

\begin{example}\label{exam-02}
 Let $(m,e)=(7,2)$. Then the code $\overline{\widehat{\C}}$ is a $[128,22,48]$ binary linear code with the weight enumerator $1+42672 z^{48}+877824 z^{56}+ 2353310 z^{64}+ 877824 z^{72}+42672 z^{80}+ z^{128}$. The dual code $\overline{\widehat{\C}}^\perp$ of $\overline{\widehat{\C}}$  has parameters $[128,106,8]$. Both $\overline{\widehat{\C}}$ and $\overline{\widehat{\C}}^\perp$ are optimal according to the tables of best known codes maintained at http: //www.codetables.de. The number of the codewords of the minimum weight $8$ in $\overline{\widehat{\C}}^\perp$ is $774192$.
\end{example}

\begin{example}\label{exam-0202}
Let $(m,e)=(10,4)$. Then the code $\overline{\widehat{\C}}$ is a $[1024,31,384]$ binary linear code with the weight enumerator
$$1+278256 z^{384}+263983104 z^{480}+ 1618960926 z^{512}+ 263983104 z^{544}+278256 z^{640}+ z^{1024}.$$
The dual code $\overline{\widehat{\C}}^\perp$ of $\overline{\widehat{\C}}$  has parameters $[1024,993,4]$. The number of the codewords of the minimum weight $4$ in $\overline{\widehat{\C}}^\perp$ is $87296$.

\end{example}

In the following, we give some $3$-designs and determine their parameters.

\begin{theorem}\label{thm-3design}
Let symbols and notation be the same as before. Let $m$ and $e$ be positive integers with $m\geq 5$ and  $1\leq e\leq m-1$, $e\notin \{\frac{m}{3},\frac{2m}{3}\}$ and $m/\gcd(m,e)$ odd. Let $\C=\C^{(e)}$ be defined by (\ref{ce}) and $\gcd(m,e)=1$.
Then the code $\overline{\widehat{\C}}$ and its dual $\overline{\widehat{\C}}^\perp$ support $3$-designs. Furthermore, the minimum weight codewords of $\overline{\widehat{\C}}$ and $\overline{\widehat{\C}}^\perp$ support simple $3$-$(q,2^{m-1}-2^{(m+1)/2},\lambda)$ designs with
\begin{eqnarray}\label{eq:numtamin1}
\lambda= \frac{(2^{m-1}-1)(2^{m-1}-2^{(m+1)/2})(2^{m-1}-2^{(m+1)/2}-1)(2^{m-1}-2^{(m+1)/2}-2)}{24 (2^m-2)})
\end{eqnarray}
and
simple $3$-$(q,8,\lambda)$ designs where
\begin{eqnarray}\label{eq:numtamin2}
\lambda= \frac{ 336\cdot  A_8 (\overline{\widehat{\C}}^\perp)}{q(q-1)(q-2)}
\end{eqnarray}
 and $A_8 (\overline{\widehat{\C}}^\perp)$ was given in (\ref{eq-A8}), respectively.
\end{theorem}

\begin{proof}

When $\gcd(m,e)=1$, from Lemma \ref{lem-extend} we have that the code $\overline{\widehat{\C}}$ has six nonzero weights, i.e.,
$w_1=2^{m-1}-2^{(m+1)/2}$,  $w_2=2^{m-1}-2^{(m-1)/2}, w_3=2^{m-1}, w_4=2^{m-1}+2^{(m-1)/2}, w_5=2^{m-1}+2^{(m+1)/2}$ and $ w_6=2^m$. It is clear that
$\left ( \mathcal P(\overline{\widehat{\C}}), \mathcal B_{i}(\overline{\widehat{\C}}) \right )$ are trivial $3$-designs for $i\in \{w_1, w_1+1,...,q-3\} \setminus \{w_1,w_2,w_3,w_4,w_5\}$. By Lemma \ref{lem-extenddual},
the minimum  distance of $\overline{\widehat{\C}}^\perp$ is $8$ when $\gcd(m,e)=1$. Thus, $\left ( \mathcal P(\overline{\widehat{\C}}^\perp), \mathcal B_{i}(\overline{\widehat{\C}}^\perp) \right )$ are trivial $3$-designs for $0\leq i \leq 5+3-1$.
From Theorem \ref{thm-designGAMtheorem} we then deduce that both $\overline{\widehat{\C}}$  and $\overline{\widehat{\C}}^\perp$ hold $3$-designs, and the minimum weight codewords in $\overline{\widehat{\C}}$  and $\overline{\widehat{\C}}^\perp$ support simple $3$-designs. Moreover, for the minimum weight $i =2^{m-1}-2^{(m+1)/2}$ in $\overline{\widehat{\C}}$ , the incidence structure $\left ( \mathcal P(\overline{\widehat{\C}}), \mathcal B_{i}(\overline{\widehat{\C}}) \right )$ is a simple $3$-$(q, 2^{m-1}-2^{(m+1)/2}, \lambda)$ design with $b$ blocks, where
\begin{eqnarray}\label{eq:b}
b=
\frac{A}{2-1}=A
\end{eqnarray}
and $A$ is the number of the the minimum weight codewords in $\overline{\widehat{\C}}$. Then the value of $\lambda$ in (\ref{eq:numtamin1}) follows from Lemma \ref{lem-extend}, Equations
(\ref{eq:bb}) and (\ref{eq:b}).

The proof of (\ref{eq:numtamin2}) is similar to that of (\ref{eq:numtamin1}).
For the minimum weight 8 in $\overline{\widehat{\C}}^\perp$ , the incidence structure $\left ( \mathcal P(\overline{\widehat{\C}})^\perp, \mathcal B_{8}(\overline{\widehat{\C}}^\perp) \right )$ is a simple $3$-$(q, 8, \lambda)$ design with
\begin{eqnarray}\label{eq:b1}
b=\frac{A_8 (\overline{\widehat{\C}}^\perp)}{2-1}=A_8 (\overline{\widehat{\C}}^\perp)
\end{eqnarray}
blocks, where $A_8 (\overline{\widehat{\C}}^\perp)$ was given in (\ref{eq-A8}). Then (\ref{eq:numtamin2}) follows from Equations (\ref{eq-A8}), (\ref{eq:b1})and (\ref{eq:bb}).
This completes the proof.
\end{proof}

\begin{example}\label{exam-1}
Let $(m,e)=(5,1)$. Then $\overline{\widehat{\C}}$  and $\overline{\widehat{\C}}^\perp$ have the same parameters $[32,16,8]$ and the same weight enumerator
$1 + 620 z^{8} +  13888 z^{12} +  36518  z^{16}+ 13888 z^{20}+ 620 z^{24}+ z^{32}.$  Both $\overline{\widehat{\C}}$  and $\overline{\widehat{\C}}^\perp$ are optimal according to the tables of best known codes  maintained at http://www.codetables.de.
The codewords of the minimum weight $8$ in $\overline{\widehat{\C}}$ (or $\overline{\widehat{\C}}^\perp$) support a $3$-$(32,8,7)$ design.
\end{example}

\section{Several shortened codes of binary cyclic codes} \label{sec-main}

In this section, we study some shortened codes of binary cyclic codes $\C^{(e)}$ and determine their parameters. Some of these shortened codes are optimal or almost optimal.

The following result can be easily obtained  and will be needed later.

\begin{lemma}\label{lem-cf}
Let $m,e,t$ be positive integers with $m\geq 5$ and  $1\leq e\leq m-1$, $e\notin \{\frac{m}{3},\frac{2m}{3}\}$ and $m/\gcd(m,e)$ odd. Let $\C=\C^{(e)}$ be defined by (\ref{ce}). 
Suppose that $\gcd(m,e)=1$, then the dual code $\C^\bot$ of $\C$ has parameters $[q-1, q-1-3m,7]$.
\end{lemma}
\begin{proof}
Note that the code $\C$ has length $q-1$ and dimension $3m$. This means that the dual code $\C^\bot$ of $\C$ has length $q-1$ and dimension $q-1-3m$. From the weight distribution in Table \ref{tab-31} and the first eight Pless power moments in (\ref{eq:PPM}), it is easily obtain that $A_7(\C^\perp)> 0$ and  $A_i(\C^\perp)=0$ for any $i\in \{1,2,3,4,5,6\}$.
The desired conclusions then follow.
\end{proof}

Let $T$  be a set of $t$ coordinate positions in $\C$ (i.e., $T$ is a $t$-subset of $\mathcal P(\C)$ ). Define
$$\Lambda _{T,w}(\C)=\{\mathrm{Supp}(\mathbf c):~wt(\mathbf{c})=w, ~\mathbf{c}\in \C~and~T\subseteq \mathrm{Supp}(\mathbf c) \}$$
and $\lambda _{T,w}(\C) = \# \Lambda _{T,w}(\C)$.

Let $\C=\C^{(e)}$ be defined by (\ref{ce}). In this section, we regard $\gf(q)^*$ as the set of the coordinate
positions $\mathcal P(\C)$  of $\C$. In the following, we will consider some shortened code $\C_T$ of $\C$ for the case $m/\gcd(m,e) $ odd and $t \geq 1$, and determine their parameter.

\subsection{Shortened codes for the case $t=1$ or $t=2$ }

In this subsection, we will consider the shortened code $\C_{T}$ and determine its parameters when $t=1$ or $t=2$.

\begin{theorem}\label{main-design1}
Let $m,e,t$ be positive integers with $1\leq e\leq m-1$, $e\notin \{\frac{m}{3},\frac{2m}{3}\}$ and $m/\gcd(m,e)$ odd. Let $\C=\C^{(e)}$ be defined by (\ref{ce}). Suppose that $\gcd(m,e)=1$,
then we have the following results.
\begin{itemize}
  \item   If $t=1$, then the shortened code $\C_{T}$ is a $[2^{m}-2, 3m-1,   2^{m-1}-2^{(m+1)/2}]$ binary linear code with the weight distribution in Table \ref{tab-t1}.
  \item   If $t=2$, then the shortened code $\C_{T}$ is a $[2^{m}-3, 3m-2,   2^{m-1}-2^{(m+1)/2}]$ binary linear code with the weight distribution in Table \ref{tab-t2}.
\end{itemize}
\end{theorem}

\begin{proof}
By Lemma \ref{lem-cf}, the minimum  distance of  $\C^\bot$ is 7. Thus, from definitions and Lemma \ref{lem:C-S-P} we have
\begin{align}\label{eq-AA}
& A_i\left ( \left (\mathcal C_{T} \right )^{\perp}  \right )=A_i\left ( \left (\mathcal C^{\perp} \right )^{T}  \right )=0
\end{align}
for any $i\in \{1,2,3,4\}$. Then the desired conclusions follow from (\ref{eq-AA}), the definitions, Lemma \ref{lem:C-S-P}, Table \ref{tab-31} and the first five Pless power moments of (\ref{eq:PPM}). This completes the proof.
\end{proof}


\begin{table}[ht]
\begin{center}
\caption{The weight distribution of $\C_T$ for $t=1$.} \label{tab-t1}
\begin{tabular}{|c|c|}
\hline
weight & multiplicity \\[2mm]
\hline
$0$ & $1$
\\[2mm]
\hline
$2^{m-1}-2^{(m+1)/2}$ & $\frac{2^{(m-13)/2} (2^m-2) (-8+ 3\cdot 2^{(m+3)/2} +2^{m+3}+2^{(3m+1)/2}   )}{3}$
\\[2mm]
\hline
$2^{m-1}-2^{(m-1)/2}$ & $\frac{2^{(m-9)/2} (8+5\cdot 2^m) (-4+2^{m+2}+2^{(3m+1)/2}   )}{3}$
\\[2mm]
\hline
$2^{m-1}$ & $\frac{(2^m-2) (16+ 3\cdot 2^{m+1}+9 \cdot 4^m )}{32}$

\\[2mm]
\hline
$2^{m-1}+2^{(m-1)/2}$ & $\frac{2^{(m-9)/2} (8+5\cdot 2^m) (4-2^{m+2}+2^{(3m+1)/2}   )}{3}$
\\[2mm]
\hline
$2^{m-1}+2^{(m+1)/2}$&  $\frac{2^{(m-13)/2} (2^m-2) (8+3 \cdot 2^{(m+3)/2}-2^{m+3}+2^{(3m+1)/2}   )}{3}$
\\[2mm]
\hline
\end{tabular}
\end{center}
\end{table}

\begin{table}[ht]
\begin{center}
\caption{The weight distribution of $\C_T$ for $t=2$.} \label{tab-t2}
\begin{tabular}{|c|c|}
\hline
weight & multiplicity \\[2mm]
\hline
$0$ & $1$
\\[1mm]
\hline
$2^{m-1}-2^{(m+1)/2}$ & $ \frac{2^{(m-15)/2} ( 32+9\cdot 2^{3(m+1)/2} -2^{m+4}-5\cdot 2^{(m+7)/2} +2^{(1+5m)/2}+3\cdot 4^{m+1} ) }{3}$
\\[2mm]
\hline
$2^{m-1}-2^{(m-1)/2}$ & $ \frac{2^{(m-11)/2} (8+5 \cdot 2^m)(-8+3\cdot 2^{m+1}+2^{(m+3)/2} + 2^{(1+3m)/2}  )    }{3}$
\\[2mm]
\hline
$2^{m-1}$ & $\frac{(2^m-4)(16+3 \cdot 2^{m+1}+9 \cdot 4^m)     }{64}$

\\[2mm]
\hline
$2^{m-1}+2^{(m-1)/2}$ & $ \frac{2^{(m-11)/2} (8+5 \cdot 2^m)(8-3\cdot 2^{m+1}+2^{(m+3)/2} + 2^{(1+3m)/2}  )    }{3}$
\\[2mm]
\hline
$2^{m-1}+2^{(m+1)/2}$ &  $ \frac{2^{(m-15)/2} ( -32+9\cdot 2^{3(m+1)/2} +2^{m+4}-5\cdot 2^{(m+7)/2} +2^{(1+5m)/2}-3\cdot 4^{m+1} ) }{3}$
\\[2mm]
\hline
\end{tabular}
\end{center}
\end{table}

\begin{example}\label{exa-des11}
Let $m=5$, $e=1$ and $T$ be a $1$-subset of $\mathcal P(\C)$. Then the shortened code $\C_{T}$  in Theorem \ref{main-design1} is a $[30,14,8]$ binary linear code with the weight enumerator $1+345z^{8}+5320z^{12}+8835z^{16}+1848z^{20}+35z^{24}$. The code $\C_{T}$ is optimal. The dual code of $\C_{T}$ has parameters $[30,16,6]$ and is almost optimal according to the tables of best known codes maintained at http: //www.codetables.de.
\end{example}

\begin{example}\label{exa-des12}
Let $m=5$, $e=2$ and $T$ be a $2$-subset of $\mathcal P(\C)$. Then the shortened code $\C_{T}$  in Theorem \ref{main-design1} is a $[29,13,8]$ linear code with the weight enumerator $1+253 z^{8}+3192 z^{12}+4123 z^{16}+616 z^{20}+7 z^{24}$. The code $\C_{T}$ is optimal. The dual code of $\C_{T}$ has parameters $[29,16,5]$ and is almost optimal according to the tables of best known codes   maintained at http://www.codetables.de.
\end{example}

\subsection{Some shortened codes for the case $t=3$ or $t=4$ }

If $t=3$ or $t=4$, Magma programs show that the weight distribution of $\C_{T}$ is very complex for any $T=\{x_1,x_2,x_3\}\subseteq \mathcal P(\C)$ or $T=\{x_1,x_2,x_3,x_4\}\subseteq \mathcal P(\C)$. Thus, it is difficult to determine their parameters in general.
In the following, we will study the $\C_{T}$ with the special set $T$ and $t=\#T=3$ (resp., $t=\#T=4$ ) in Theorem \ref{main-t33} (resp., Theorem \ref{main-t34} ). In order to determine the parameters of $\C_{T}$, we need the results in the following two lemmas.

\begin{lemma}\label{A4gel}
Let $m\geq 2$ be an integer, $q=2^m$ and $(b,c)\in (\gf(q)^*)^2$.  Let $N_{(b,c)}$ be the number of $\{x,y,z,u\}\subseteq \gf(q)$ satisfying the system of equations
\begin{eqnarray}\label{numA4gel}
\left\{
\begin{array}{ll}
x+y+z+u=0&\\[2mm]
x^{3}+y^{3}+z^{3}+u^{3}=b&\\[2mm]
x^{5}+y^{5}+z^{5}+u^{5}=c
\end{array}
 \right..
\end{eqnarray}
Suppose that the cubic equation
\begin{eqnarray}\label{namta3}
\lambda^3+\frac{c}{b} \lambda+b=0,
\end{eqnarray}
on $\lambda$ over $\gf(q)$ has three pairwise distinct solutions denoted by $\lambda_1$, $\lambda_2$ and $\lambda_3$,
then $\{x,y,z,u\}\subseteq \gf(q)$ satisfying (\ref{numA4gel}) iff
  $$
  \{x,y,z,u\}=\{\lambda, \lambda+\lambda_1, \lambda+\lambda_2, \lambda+\lambda_3 \}
  $$
  where $\lambda \in \gf(q)$. Furthermore, $N_{(b,c)}=2^{m-2}$.
\end{lemma}

\begin{proof}
By the first equation of (\ref{numA4gel}), it is clear that
$$y+z+u=x,~~y^2+z^2+u^2=x^2,~~y^4+z^4+u^4=x^4.
$$
Further, we have
$$
(y+x)^3+(z+x)^3+(u+x)^3=y^3+z^3+u^3+x^3+x^2(y+z+u)+x(y^2+z^2+u^2)=y^3+z^3+u^3+x^3
$$
and
$$
(y+x)^5+(z+x)^5+(u+x)^5=y^5+z^5+u^5+x^5+x^4(y+z+u)+x(y^4+z^4+u^4)=y^5+z^5+u^5+x^5.
$$
Denote $Y=y+x$, $Z=z+x$ and $U=u+x$. Then the system (\ref{numA4gel}) can be reduced to
\begin{eqnarray}\label{numA4gel2}
\left\{
\begin{array}{ll}
Y+Z+U=0&\\[2mm]
Y^3+Z^3+U^3=b&\\[2mm]
Y^5+Z^5+U^5=c
\end{array}
 \right..
\end{eqnarray}
Note that
\begin{eqnarray}\label{Y1}
Y^3+Z^3+U^3+3YZU=(Y+Z+U)(Y^2+Z^2+U^2+YZ+ZU+UY)
\end{eqnarray}
and
\begin{eqnarray}\label{Y2}
Y^5+Z^5+U^5=\sigma_1 (Y^4+Z^4+U^4)+\sigma_2 (Y^3+Z^3+U^3)+\sigma_3 (Y^2+Z^2+U^2),
\end{eqnarray}
where $\sigma_i$ is the elementary symmetric polynomial, i.e., $\sigma_1=X+Y+Z$, $\sigma_2=YZ+ZU+UY$ and $\sigma_3=YZU$.
By the first two equations of (\ref{numA4gel2}), it is clear that
\begin{eqnarray}\label{xyz0}
\sigma_1=0.
\end{eqnarray}
and from Equation (\ref{Y1}) we have
\begin{eqnarray}\label{xyz1}
\sigma_3=YZU=b.
\end{eqnarray}
Meanwhile, by (\ref{numA4gel2}), from Equation (\ref{Y2}) we have $c=b\sigma_2$ which means that
\begin{eqnarray}\label{xyz2}
\sigma_2=\frac{c}{b}.
\end{eqnarray}
Combining Equations (\ref{xyz0}), (\ref{xyz1}) and (\ref{xyz2}), then $Y,~Z,~U$ are the three pairwise distinct solutions of (\ref{namta3}).
Without the loss of generality, we assume that $\lambda_1=Y, ~\lambda_2=Z, ~\lambda_3=U.$ Then
$$x=\lambda, ~y=\lambda+\lambda_1,~ z=\lambda+\lambda_2, ~u=\lambda+\lambda_3,$$
where $\lambda \in \gf(q)$. Conversely, it is not hard to verify that
$$\{x=\lambda, ~y=\lambda+\lambda_1,~ z=\lambda+\lambda_2, ~u=\lambda+\lambda_3 \}$$
satisfies (\ref{numA4gel}) for any $\lambda \in \gf(q)$. Thus $N_{(b,c)}=2^{m-2}$. This completes the proof.
\end{proof}

\begin{lemma}\label{A4N0}
Let $m\geq 3$ be odd, $q=2^m$ and $T=\{x_1,x_2,x_3\} \subseteq  \gf(q)^*$. Define $a=x_1+x_2+x_3$, $b=x_1^{3}+x_2^{3}+x_3^{3}$ and $c=x_1^{5}+x_2^{5}+x_3^{5} $.
Let $N(a)$ be the number of $\{x,y,z,u\}\subseteq \gf(q)$ satisfying the system of equations
\begin{eqnarray}\label{numA40}
\left\{
\begin{array}{ll}
x+y+z+u=a&\\[2mm]
x^{3}+y^{3}+z^{3}+u^{3}=b&\\[2mm]
x^{5}+y^{5}+z^{5}+u^{5}=c
 \end{array}
 \right.
\end{eqnarray}
such that $x_1,x_2,x_3,x,y,z,u$ are pairwise distinct. Then we have the following results.
\begin{itemize}
  \item If $a \neq 0$, then
\begin{eqnarray}\label{eqn-solution-an0}
N(a)&=&
\left\{
\begin{array}{ll}
0              &   \mbox{ if } \mu=1 \\
\frac{1}{24} \left(2^m-5+3G_m(\mu+1,\mu+1)\right)+\frac{1}{12}(-1)^{\varepsilon+1}    &   \mbox{ if } \mu \neq 1 \\
~~~~~~~~\times \left(K_m(\mu+1)+C_m(1,(\mu+1)^{1/3})-3\right) -1   &    ~\\
\end{array} \right.
\end{eqnarray}
where $\varepsilon$ and $\mu$ were defined in Lemma \ref{lem-solution}.

\item If $a = 0$, then $bc\neq 0$ and
  $N(a)=2^{m-2}-1$.
\end{itemize}
\end{lemma}

\begin{proof}
The desired conclusions then follow from definitions and Lemmas \ref{lem-solution} and  \ref{A4gel}.
\end{proof}

\begin{theorem}\label{main-t33}
Let $m\geq 5$ be odd, $q=2^m$ and $\C=\C^{(1)}$ be defined by (\ref{ce}). Let $T=\{x_1,x_2,x_3\} \subseteq  \gf(q)^*$. Denote $a=x_1+x_2+x_3$.
Then the shortened code $\C_{T}$ is a $[2^{m}-4, 3m-3,   2^{m-1}-2^{(m+1)/2}]$ binary code with the weight distribution in Table \ref{tab-t3}, where $N=N(a)$ was given in Lemma \ref{A4N0}.
\end{theorem}

\begin{table}[ht]
\begin{center}
\caption{The weight distribution of $\C_T$ for $t=3$.} \label{tab-t3}
\begin{tabular}{|c|c|}
\hline
weight & multiplicity \\[2mm]
\hline
$0$ & $1$
\\[2mm]
\hline
$2^{m-1}-2^{(m+1)/2}$ & $ \frac{2^{( m-17)/2} (64 + 19 \cdot 2^{(3 (1 + m))/2} + 2^{5 + m} + 2^{
    1/2 (1 + 5 m)} + 4^{2 + m} + 2^{(7 + m)/2} (-7 + 3 N))}{3}$
\\[2mm]
\hline
$2^{m-1}-2^{(m-1)/2}$ & $ \frac{2^{(m-13)/2} (-128 + 17 \cdot 2^{(3 (1 + m))/2} - 2^{4 + m} +
   5 \cdot 2^{3 + 2 m} + 5 \cdot  2^{(1 + 5 m)/2} + 2^{(7 + m)/2} (1 - 3 N))}{3}$
\\[2mm]
\hline
$2^{m-1}$ &   $ -1 + 9 \cdot 2^{-7 + 3 m} - 29 \cdot 4^{m-3} + 2^{m-4} (1 + 3 N) $

\\[2mm]
\hline
$2^{m-1}+2^{(m-1)/2}$ & $ \frac{2^{(m-13)/2} (128 + 17 \cdot 2^{(3 (1 + m))/2} +2^{4 + m} -
   5 \cdot 2^{3 + 2 m} + 5 \cdot  2^{(1 + 5 m)/2} + 2^{(7 + m)/2} (1 - 3 N))}{3}$
\\[2mm]
\hline
$2^{m-1}+2^{(m+1)/2}$ &  $ \frac{2^{( m-17)/2} (-64 + 19 \cdot 2^{(3 (1 + m))/2} - 2^{5 + m} + 2^{
    (1 + 5 m)/2} - 4^{2 + m} + 2^{(7 + m)/2} (-7 + 3 N))}{3}$
\\[2mm]
\hline
\end{tabular}
\end{center}
\end{table}

\begin{proof}
Since the minimum  distance of  $\C^\bot$ is 7, the minimum  distance of $(\mathcal C_{T})^{\perp}$ is at least $4$. This means that
\begin{align}\label{eq-AA1}
& A_1\left ( \left (\mathcal C_{T} \right )^{\perp}  \right )= A_2\left ( \left (\mathcal C_{T} \right )^{\perp}  \right )=A_3 \left ( \left (\mathcal C_{T} \right )^{\perp}  \right )=0.
\end{align}
Moreover, from definitions we have
\begin{align}\label{eq-AA2}
& A_4\left ( \left (\mathcal C_{T} \right )^{\perp}  \right )=N(a),
\end{align}
where $N(a)$ was given in Lemma \ref{A4N0}. Then the desired conclusions follow from definitions, (\ref{eq-AA1}), (\ref{eq-AA2}), Table \ref{tab-31} and the first five Pless power moments of (\ref{eq:PPM}). This completes the proof.
\end{proof}

\begin{example}\label{exa-310}
Let $m=5$, $w$ be a primitive element of $\gf(2^5)$ with minimal polynomial $w^5+w^2+1=0$  and $T=\{w,~w^2,~w^3\}$. Then $a=w+w^2+w^3=w^{12} \neq 0$, $b=w^3+w^6+w^9=w^{21}$, $c=w^5+w^{10}+w^{15}=w^{27}$,
$$\frac{c}{a^5}+\frac{b^2}{a^6}+\frac{b}{a^3}=w^4 \neq 1$$
and $N(a)=1$.
The shortened code $\C_{T}$ in Theorem \ref{main-t33} is a $[ 28, 12, 8]$ binary code with the weight enumerator
$$1+ 183 z^{8}+ 1872 z^{12}+ 1847 z^{16}+ 192z^{20}+  z^{24}$$
and its dual has parameters $[28,16,4]$. The code $\C_{T}$ is optimal according to the tables of best known codes maintained at http: //www.codetables.de.
\end{example}

\begin{example}\label{exa-311}
Let $m=5$, $w$ be a primitive element of $\gf(2^5)$ with minimal polynomial $w^5+w^2+1=0$  and $T=\{w,~w^3,~w^6\}$. Then $a=w+w^3+w^6=0$, $b=w^3+w^9+w^{18}=w^{10}\neq 0$, $c=w^5+w^{15}+w^{30}=w^{3}\neq 0$ and
$N(a)=7$.
The shortened code $\C_{T}$ in Theorem \ref{main-t33} is a $[ 28, 12, 8]$ binary code with the weight enumerator
$$1+ 189 z^{8}+ 1848 z^{12}+ 1883 z^{16}+ 168z^{20}+  7 z^{24}$$
and its dual has parameters $[28,16,4]$. The code $\C_{T}$ is optimal according to the tables of best known codes maintained at http: //www.codetables.de.
\end{example}

\begin{example}\label{exa-312}
Let $m=7$, $w$ be a primitive element of $\gf(2^7)$ with minimal polynomial $w^7+w+1=0$  and $T=\{w,~w^{20},~w^{30}\}$. Then $a=w+w^{20}+w^{30}=0$
and $N(a)=31$.
The shortened code $\C_{T}$ in Theorem \ref{main-t33} is a $[ 124, 18, 48]$ binary code with the weight enumerator
$$1+ 6430 z^{48}+ 84240 z^{56}+ 140783 z^{64}+ 29808 z^{72}+  882 z^{80}$$
and its dual has parameters $[124,106,4]$. The code $\C_{T}$ is optimal according to the tables of best known codes maintained at http: //www.codetables.de.
\end{example}

To prove Theorem \ref{main-t34}, we need the results in the following two lemmas.The former was documented in \cite{HK2010} and the latter is given in Lemma \ref{lem-solu47}

\begin{lemma}\label{lem-degree3}\cite{HK2010}
Let $\delta_1, \delta_2, \delta_3 \in \gf(q)$ with $\delta_1 ^2 \neq \delta_2$ and $\delta_3 \neq \delta_1\delta_2$. Denote $\delta=(\delta_2+\delta_1^2)^3/(\delta_3+\delta_1\delta_2)^2$ and the cubic
equation
\begin{eqnarray}\label{eq-degree3}
x^3+\delta_1 x^2 +\delta_2 x+\delta_3=0.
\end{eqnarray}
Then the following results hold.
\begin{itemize}
  \item $\tr(\delta+1)=1$ if and only if  Equation (\ref{eq-degree3}) has  a unique solution $x \in \gf(q)$.
  \item If $\tr(\delta+1)=0$, then Equation (\ref{eq-degree3}) has zero solution or three distinct solutions in $\gf(q)$.
\end{itemize}
\end{lemma}

\begin{lemma}\label{lem-solu47}
Let $m\geq 5$ be odd, $q=2^m$ and $ \{ x_1, x_2, x_3, x_4\}$ be a $4$-subset of  $ \gf(q)^* $. Denote $S_i=x_1^i+x_2^i+x_3^i+x_4^i$.
Let $\overline{N}$ be the number of $\{x,y,z\}\subseteq \gf(q)^*$ satisfying the system of equations

\begin{eqnarray}\label{eqx3}
\left\{
   \begin{array}{cccl}
    & x~+y~+z~=S_1  \\ [2mm]
    & x^3+y^3+z^3=S_3  \\ [2mm]
    & x^{5} +y^{5}+z^{5}=S_{5}
   \end{array}
\right.
\end{eqnarray}
with $\# \{x_1, x_2, x_3, x_4, x,y,z\}=7.$
Then $S_3+S_1^3 \neq 0$ and $S_5+S_1^5 \neq 0$. Furthermore, $\overline{N}=0$ if $\tr\left (\frac{(S_5+S_1 ^5)^3}{(S_3+S_1^3)^5}+1 \right )=1$, $\overline{N}=0$ or $1$ otherwise.
\end{lemma}

\begin{proof}
Substituting
\begin{eqnarray}\label{eqx0}
x= x + S_1,~y= y + S_1~,z= z + S_1
\end{eqnarray}
into the system (\ref{eqx3}) leads to
\begin{eqnarray}\label{eqx4}
\left\{
   \begin{array}{cccl}
    & x~+y~+z~=0  \\ [2mm]
    & x^3+y^3+z^3=S_3+S_1^3  \\ [2mm]
    & x^{5} +y^{5}+z^{5}=S_5+S_1^5
   \end{array}
\right.
\end{eqnarray}
This means that $S_3+S_1^3 \neq 0$ and $S_5+S_1^5 \neq 0$, since the codes with zero sets $\{1, 2^k + 1\}$ where $\gcd(k, m) = 1$ have minimum distance five.
Further, substituting $z= x + y$ into the last two equations of (\ref{eqx4}) leads to
\begin{eqnarray}\label{eqx5}
\left\{
   \begin{array}{ccl}
    & x^2 y+y^2 x=S_3+S_1^3  \\ [2mm]
    & x^4 y+ +y^4 x=S_5+S_1^5
   \end{array}
\right.
\end{eqnarray}
and replacing $y$ with $xy$ into (\ref{eqx5}) leads to
\begin{eqnarray}\label{eqx6}
\left\{
   \begin{array}{ccl}
    & y+y^2=x^{-3}(S_3+S_1^3)  \\ [2mm]
    & y+y^4 =x^{-5}(S_5+S_1^5)
   \end{array}
\right.
\end{eqnarray}
Note that
\begin{eqnarray}\label{eqx7}
(y+y^2)+(y+y^2)^2=y+y^4.
\end{eqnarray}
Applying (\ref{eqx6}) to (\ref{eqx7}), we have
\begin{eqnarray}\label{eqx8}
x^3+(S_3+S_1^3)^{-1}(S_5+S_1^5)x+(S_3+S_1^3)=0.
\end{eqnarray}
Next we distinguish two cases as follows:
\begin{itemize}
  \item When $\tr\left (\frac{(S_5+S_1 ^5)^3}{(S_3+S_1^3)^5}+1 \right )=1$, from Lemma \ref{lem-degree3} we have (\ref{eqx8}) has a unique solution $\xi \in \gf(q)^*$.
Since $x,y,z$ of (\ref{eqx3}) have symmetrical property, from (\ref{eqx0}) we have $x+S_1=y+S_1=z+S_1=\xi $ which means that
 $(\xi+S_1,\xi+S_1,\xi+S_1)$ is the unique solution of (\ref{eqx3}). Thus, $\overline{N}=0$.
  \item When $\tr\left (\frac{(S_5+S_1 ^5)^3}{(S_3+S_1^3)^5}+1 \right )=0$, from Lemma \ref{lem-degree3} we have (\ref{eqx8}) has zero solution or three distinct solutions. If (\ref{eqx8}) has zero solution, then $\overline{N}=0$. If (\ref{eqx8}) has three distinct solutions, i.e., $\xi_1$, $\xi_2$, and $\xi_3$,  then there exists a unique set $\{x,y,z\}=\{\xi_1+S_1,\xi_2+S_1,,\xi_3+S_1\}$ satisfying (\ref{eqx3}) and thus $\overline{N}=1$.
\end{itemize}
This completes the proof.
\end{proof}

\begin{theorem}\label{main-t34}
Let $m\geq 5$ be odd, $q=2^m$ and $\C=\C^{(1)}$ be defined by (\ref{ce}). Let $T=\{x_1,x_2,x_3,x_4\} \subseteq  \gf(q)^*$. Denote $a_i=(\sum_{j=1}^4 x_j) -x_i$ and $S_i=x_1^i+x_2^i+x_3^i+x_4^i$.
Then the shortened code $\C_{T}$ is a $[2^{m}-5, 3m-4, 2^{m-1}-2^{(m+1)/2}]$ binary code with the weight distribution in Table \ref{tab-t41}, where
$N=\sum_{i=1}^4 N(a_i) -4\cdot \overline{N} +N(S_1)$, $\overline{N}$ was given in Lemma \ref{lem-solu47} and $N(\cdot)$ was defined in Lemma \ref{A4N0}.
\end{theorem}
\begin{table}[ht]
\begin{center}
\caption{The weight distribution of $\C_T$ for $t=4$.} \label{tab-t41}
\begin{tabular}{|c|c|}
\hline
weight & multiplicity \\[2mm]
\hline
$0$ & $1$
\\[2mm]
\hline
$2^{m-1}-2^{\frac{m+1}{2}}$ & $ \frac{2^{( m-19)/2} (128+31 \cdot 2^{3(m+1)/2}+2^{(5m+1)/2}+5\cdot 4^{m+1}+2^{m+3}\cdot(17+3\overline{N})+2^{(m+7)/2} (15\overline{N}+3N-1)    )}{3} $

\\[2mm]
\hline
$2^{m-1}-2^{\frac{m-1}{2}}$ & $ \frac{2^{\frac{m-15}{2}} \cdot (-256 +29  \cdot 2^{3(m+1)/2} +25 \cdot 2^{1+2m}+5\cdot 2^{(5m+1)/2}-2^{m+2}(17+3\overline{N}) -2^{(m+7)/2} (11+15 \overline{N}+3N)       )}{3}$
\\[2mm]
\hline

$2^{m-1}$ & $-1+9\cdot 2^{3m-8}-49 \cdot 2^{2m-7}+3 \cdot 2^{m-5}\cdot (5+5\overline{N}+N)$

\\[2mm]
\hline
$2^{m-1}+2^{\frac{m-1}{2}}$ & $ \frac{2^{\frac{m-15}{2}} \cdot (256 +29  \cdot 2^{3(m+1)/2} -25 \cdot 2^{1+2m}+5\cdot 2^{(5m+1)/2}+2^{m+2}(17+3\overline{N}) -2^{(m+7)/2} (11+15 \overline{N}+3N)       )}{3}$
\\[2mm]
\hline
$2^{m-1}+2^{\frac{m+1}{2}}$ &  $ \frac{2^{( m-19)/2} (-128+31 \cdot 2^{3(m+1)/2}+2^{(5m+1)/2}-5\cdot 4^{m+1}-2^{m+3}\cdot(17+3\overline{N})+2^{(m+7)/2} (15\overline{N}+3N-1)    )}{3} $
\\[2mm]
\hline
\end{tabular}
\end{center}
\end{table}

\begin{proof}
By $\#T=4$ and the minimum distance of  $C^{\perp}$ is 7, then the minimum distance of  $\left (\mathcal C^{\perp} \right )^{T}$ is at least 3. This means that
\begin{align}\label{eq-4A12}
A_1\left ( \left (\mathcal C^{\perp} \right )^{T}  \right )=A_2\left ( \left (\mathcal C^{\perp} \right )^{T}  \right )=0.~
\end{align}
Further, from the definition of $T$ we have
\begin{align}\label{eq:A3}
\lambda _{T,7}(\C^\bot)=\overline{N}
\end{align}
where $\overline{N}$ was given in Lemma \ref{lem-solu47}.
Therefore,
\begin{align}\label{eq-4A3}
A_3\left ( \left (\mathcal C^{\perp} \right )^{T}  \right )= \overline{N}.
\end{align}
From definitions and Lemma \ref{A4N0} we get
\begin{align}\label{eq-84}
\lambda _{T,8}(\C^\bot)=N(S_1),
\end{align}
where $N(S_1)$ was defined in Lemma \ref{A4N0}.
Note that there exist only four $3$-subsets of $T$, i.e,
$$T_1=\{x_2,x_3,x_4\}, T_2=\{x_1,x_3,x_4\}, T_3=\{x_1,x_2,x_4\}, T_4=\{x_1,x_2,x_3\}$$
and
\begin{align}\label{eq-00}
\lambda _{T_i,7}(\C^\bot)=N(a_i),
\end{align}
where $N(\cdot)$ was defined in Lemma \ref{A4N0}.
By definitions we have
\begin{align}\label{eq-4A4}
A_4\left ( \left (\mathcal C^{\perp} \right )^{T}  \right )&= \sum_{i=1}^{4}\left( A_4\left ( \left (\mathcal C^{\perp} \right )^{T_i}  \right )- \lambda _{T,7}(\C^\bot)    \right)+
\lambda _{T,8}(\C^\bot)    \nonumber \\
&= \sum_{i=1}^{4}\left(\lambda _{T_i,7}(\C^\bot) - \lambda _{T,7}(\C^\bot)    \right)+ \lambda _{T,8}(\C^\bot)
\end{align}
Combining Equations (\ref{eq:A3}),  (\ref{eq-84}) and (\ref{eq-00}) with Equation (\ref{eq-4A4}) yields
\begin{align}\label{eq-even4A4}
A_4\left ( \left (\mathcal C^{\perp} \right )^{T}  \right )= \sum_{i=1}^4 N(a_i)-4 \cdot \overline{N} +N(S_1),
\end{align}
where $N(\cdot)$ was defined in Lemma \ref{A4N0}.

Note that the shortened code $\C_{T}$ has length $q-1-4=2^m-5$ and dimension $3m-4$ due to $\#T=4$ and Lemma \ref{lem:C-S-P}.
By definition and the weight distribution in Table \ref{tab-31}, we have that
$A_i\left (\C_T \right )=0$ for $i\not \in \{0, i_1, i_2, i_3,i_4,i_5\} $, where $i_1=2^{m-1}-2^{(m+1)/2}$, $i_2=2^{m-1}-2^{(m-1)/2}$
, $i_3=2^{m-1}$, $i_4=2^{m-1}+2^{(m-1)/2}$  and $i_5=2^{m-1}+2^{(m+1)/2}$.
Using Lemma \ref{lem:C-S-P} and Equations (\ref{eq-4A12}), (\ref{eq-4A3}) and (\ref{eq-even4A4}) and applying the first five Pless power moments of (\ref{eq:PPM}) yields the weight distribution in Table \ref{tab-t41}. This completes the proof.
\end{proof}

\begin{corollary}
In Theorem \ref{main-t34}, if $\tr\left (\frac{(S_5+S_1 ^5)^3}{(S_3+S_1^3)^5}+1 \right )=1$, from Lemma \ref{lem-solu47} we have $\overline{N}=0$. Thus, the shortened code $\C_{T}$ is a $[2^{m}-5, 3m-4, 2^{m-1}-2^{(m+1)/2}]$ binary code with the weight distribution in Table \ref{tab-t42}.
\end{corollary}

\begin{table}[ht]
\begin{center}
\caption{The weight distribution of $\C_T$ for $\tr\left (\frac{(S_5+S_1 ^5)^3}{(S_3+S_1^3)^5}+1 \right )=1$.} \label{tab-t42}
\begin{tabular}{|c|c|}
\hline
weight & multiplicity \\[2mm]
\hline
$0$ & $1$
\\[2mm]
\hline
$2^{m-1}-2^{(m+1)/2}$ & $ \frac{2^{( m-19)/2} (128-2^{(m+7)/2}+17\cdot 2^{m+3}+31\cdot 2^{(3m+3)/2}+5\cdot 2^{2+2m}+2^{(5m+1)/2}+3 \cdot 2^{(m+7)/2} \cdot N)}{3} $

\\[2mm]
\hline
$2^{m-1}-2^{(m-1)/2}$ & $ \frac{2^{\frac{m-15}{2}} \cdot (-256 -11  \cdot 2^{(m+7)/2} -17\cdot 2^{m+2}+29\cdot 2^{(3m+3)/2}+25 \cdot 2^{1+2m}+5 \cdot 2^{(5m+1)/2}-3 \cdot 2^{(m+7)/2} \cdot N)}{3}$
\\[2mm]
\hline
$2^{m-1}$ & $-1+9\cdot 2^{3m-8}+15 \cdot 2^{m-5}-49 \cdot 2^{2m-7}+3 \cdot 2^{m-5}\cdot N$

\\[2mm]
\hline
$2^{m-1}+2^{(m-1)/2}$ & $ \frac{2^{\frac{m-15}{2}} \cdot (256 -11  \cdot 2^{(m+7)/2} +17\cdot 2^{m+2}+29\cdot 2^{(3m+3)/2}-25 \cdot 2^{1+2m}+5 \cdot 2^{(5m+1)/2}-3 \cdot 2^{(m+7)/2} \cdot N)}{3}$
\\[2mm]
\hline
$2^{m-1}+2^{(m+1)/2}$ &  $ \frac{2^{( m-19)/2} (-128-2^{(m+7)/2}-17\cdot 2^{m+3}+31\cdot 2^{(3m+3)/2}-5\cdot 2^{2+2m}+2^{(5m+1)/2}+3 \cdot 2^{(m+7)/2} \cdot N)}{3} $
\\[2mm]
\hline
\end{tabular}
\end{center}
\end{table}

\begin{example}\label{exa-41}
Let $m=5$, $w$ be a primitive element of $\gf(2^5)$ with minimal polynomial $w^5+w^2+1=0$  and $T=\{w^2,~w^4,~w^5,~w^8\}$. Then $S_1=w^{13} $, $S_3=w^{5}$, $S_5=w^{24}$,
$\tr\left (\frac{(S_5+S_1 ^5)^3}{(S_3+S_1^3)^5}+1 \right )=1$, $\overline{N}=0$
and $N=5$. The shortened code $\C_{T}$ in Theorem \ref{main-t34} is a $[ 27, 11, 8]$ binary code with the weight enumerator
$$1+ 130 z^{8}+ 1072 z^{12}+ 789 z^{16}+ 56 z^{20}$$
and its dual has parameters $[27,16,4]$. The code $\C_{T}$ is optimal according to the tables of best known codes maintained at http: //www.codetables.de.
\end{example}

\begin{example}\label{exa-42}
Let $m=7$, $w$ be a primitive element of $\gf(2^7)$ with minimal polynomial $w^7+w+1=0$  and $T=\{w^2,~w^3,~w^6,~w^7\}$. Then $S_1=w^{37} $, $S_3=w^{67}$, $S_5=w^{26}$,
$\tr\left (\frac{(S_5+S_1 ^5)^3}{(S_3+S_1^3)^5}+1 \right )=1$, $\overline{N}=0$
and $N=27$. The shortened code $\C_{T}$ in Theorem \ref{main-t34} is a $[ 123, 17, 48]$ binary code with the weight enumerator
$$1+ 3878 z^{48}+ 46416 z^{56}+ 67839 z^{64}+ 12656 z^{72}+  282 z^{80}$$
and its dual has parameters $[123,106,4]$. The code $\C_{T}$ is optimal according to the tables of best known codes maintained at http: //www.codetables.de.
\end{example}

\begin{rem}
In Theorem \ref{main-t34}, if $\tr\left (\frac{(S_5+S_1 ^5)^3}{(S_3+S_1^3)^5}+1 \right )=0$, the value of $\lambda _{T,7}(\C^\bot)$ is $0$ or $1$. Note that the value of $\lambda _{T,7}(\C^\bot)$ is related to the number of solutions of (\ref{eqx8}). If (\ref{eqx8}) has zero solution, then $\overline{N}=\lambda _{T,7}(\C^\bot)=0$. If (\ref{eqx8}) has three distinct solutions, then $\overline{N}=\lambda _{T,7}(\C^\bot)=1$.
\end{rem}

\begin{example}\label{exa-43}
Let $m=5$, $w$ be a primitive element of $\gf(2^5)$ with minimal polynomial $w^5+w^2+1=0$  and $T=\{w,~w^3,~w^6,~w^7\}$. Then $S_1=w^{7} $, $S_3=w^{29}$, $S_5=w^{21}$,
$\tr\left (\frac{(S_5+S_1 ^5)^3}{(S_3+S_1^3)^5}+1 \right )=0$, $\overline{N}=1$
and $N=6$. The shortened code $\C_{T}$ in Theorem \ref{main-t34} is a $[ 27, 11, 8]$ binary code with the weight enumerator
$$1+ 135 z^{8}+ 1056 z^{12}+ 807 z^{16}+ 48 z^{20}+ z^{24}$$
and its dual has parameters $[27,16,3]$. The code $\C_{T}$ is optimal according to the tables of best known codes maintained at http: //www.codetables.de.
\end{example}

\begin{example}\label{exa-44}
Let $m=5$, $w$ be a primitive element of $\gf(2^5)$ with minimal polynomial $w^5+w^2+1=0$   and $T=\{w^2,~w^3,~w^6,~w^7\}$. Then $S_1=w^{30} $, $S_3=w^{27}$, $S_5=w^{20}$,
$\tr\left (\frac{(S_5+S_1 ^5)^3}{(S_3+S_1^3)^5}+1 \right )=0$, $\overline{N}=0$
and $N=5$. Then the parameters of the shortened code $\C_{T}$ are the same as the parameters of $\C_{T}$ in Example \ref{exa-41}.
\end{example}

\section{Summary and concluding remarks}\label{sec-summary}
In this paper, we mainly investigated the extended codes of the augmented codes from a class of cyclic codes with three zeros. The results shown that those extended codes and their dual codes held $3$-designs. Meanwhile, we also studied some shortened codes of the cyclic codes and completely determined their weight distributions. It was shown
that the parameters of the presented shortened codes are new, and some of those
codes are optimal in the sense that their parameters meet certain bound of linear
codes. We remark that this paper does not consider the shortened codes of $\C^{(e)}$ for $ m$ being even, since it is hard to determine the weight distributions of these shortened codes $(\C^{(e)})_T$ for $\#T \geq 2$ using our methods.
\section*{Acknowledgements}

The authors are very grateful to the reviewers and the Editor, for their comments and suggestions that improved the presentation and quality of this paper.
C. Xiang's research was supported by the National Natural Science Foundation of China under grant numbers 12171162 and 11971175, and the Basic Research Project of Science and Technology Plan of Guangzhou city of China under grant number 202102020888. C. Tang's research was supported by the National Natural Science Foundation of China under grant number 12231015, the Sichuan Provincial Youth Science and Technology Fund under grant number 2022JDJQ0041 and  the Innovation Team Funds of China West Normal University under grant number KCXTD2022-5.


\end{document}